\documentclass{article}

\usepackage{amsmath,amssymb,graphicx,subfigure}
\usepackage{epsfig}
\usepackage{epic}
\usepackage{eepic}
\usepackage{graphics}

\def\debut{\begin{itemize}\item[{\bf [[}]\small}
\def\term{\hfill {\bf]]} \end{itemize} }

\newcommand{\ds}{\displaystyle}
\newcommand{\ra}{\rightarrow}

\newcommand{\et}{\mbox{\ and\ }} 
\newcommand{\vs}{\vspace{11pt}}
\newcommand{\tq}{\ | \ }
\newcommand{\Aa}{\alpha}
\newcommand{\Bb}{\beta}
\newcommand{\Cc}{\gamma}
\newcommand{\Dd}{\lambda}

\newtheorem{claim}{Claim}
\newtheorem{theorem}{Theorem}
\newtheorem{proposition}[theorem]{Proposition}

\newtheorem{conjecture}[theorem]{Conjecture}

\newtheorem{corollary}[theorem]{Corollary} 
\newtheorem{lemma}[theorem]{Lemma}
\newtheorem{remark}[theorem]{Remark}
\newenvironment{proof}{\par \noindent {\bf Proof}.\ }{\hfill$\Box$ \par \vspace{11pt}}

\newcommand\pas[1]{}

\begin{document}

\title{\bf WDM and Directed Star Arboricity}

\author{Omid Amini~\thanks{CNRS-DMA, \'Ecole Normale Sup\'erieure, Paris, France. This work was done while this author was PhD student at \'Ecole Polytechnique and Projet Mascotte, INRIA Sophia-Antipolis, France. {\tt omid.amini@m4x.org}} $^{,1}$ \and Fr\'ed\'eric Havet~\thanks{Projet Mascotte, CNRS/INRIA/UNSA,
INRIA Sophia-Antipolis, France. {\tt fhavet@sophia.inria.fr}} $^{,1,2}$ \and Florian Huc~\thanks{Universit\'e de Gen\`eve, Switzerland. This work was done while this author was PhD student in Projet Mascotte, INRIA Sophia-Antipolis, France. {\tt florian.huc@unige.ch}} $^{,1}$ \and St\'ephan Thomass\'e~\thanks{LIRMM, Montpellier, France. {\tt thomasse@lirmm.fr}\newline
$^{\tiny 1}$ These authors were partially supported by the European project AEOLUS. \newline
$^{\tiny 2}$ Supported by the ANR Blanc AGAPE.} $^{,2}$
}

\date{ } 

\maketitle

\begin{abstract}

A digraph is $m$-labelled if every arc is labelled by an integer
in $\{1, \dots ,m\}$.  Motivated by wavelength assignment for
multicasts in optical networks, we introduce and study $n$-fibre colourings of
labelled digraphs. These are colourings of the arcs of $D$ such that
at each vertex $v$, and for each colour $\alpha$,
$in(v,\alpha)+out(v,\alpha)\leq n$ with $in(v,\alpha)$ the number
of arcs coloured $\alpha$ entering $v$ and $out(v,\alpha)$ the
number of labels $l$ such that there is at least one arc of label $l$
leaving $v$ and coloured with $\alpha$.  The problem is to find the
minimum number of colours $\lambda_n(D)$ such that the $m$-labelled
digraph $D$ has an $n$-fibre colouring.  In the particular case when
$D$ is $1$-labelled, $\lambda_1(D)$ is called the directed star
arboricity of $D$, and is denoted by $dst(D)$.  We first show that
$dst(D)\leq 2\Delta^-(D)+1$, and conjecture that if $\Delta^-(D)\geq 2$,
then $dst(D)\leq 2\Delta^-(D)$. We also prove that for a subcubic digraph $D$,
then $dst(D)\leq 3$, and that if $\Delta^+(D), \Delta^-(D)\leq 2$, then
$dst(D)\leq 4$.  Finally, we study $\lambda_n(m,k)=\max\{\lambda_n(D)
\tq D \mbox{ is $m$-labelled} \et \Delta^-(D)\leq k\}$.  We show that
if $m\geq n$, then $\ds \left\lceil\frac{m}{n}\left\lceil
\frac{k}{n}\right\rceil +
\frac{k}{n} \right\rceil\leq  \lambda_n(m,k) \leq\left\lceil\frac{m}{n}\left\lceil \frac{k}{n}\right\rceil +
\frac{k}{n} \right\rceil + C \frac{m^2\log k}{n}$ for some constant $C$.
We conjecture that the lower bound should be the right value of $\lambda_n(m,k)$.

\end{abstract}

\section{Introduction}

The origin of this paper is the study of wavelength assignment for multicasts 
in star networks. We are given a star network in which a central node is connected by optical fibres to a set 
of nodes $V$.  The nodes of $V$ communicates together using a technology called WDM ({\it wavelength-division multiplexing}), which allows to send different signals at the same time through the same fibre but on different wavelengths. The central node or hub is an all-optical transmitter
which can redirect a signal arriving from a node on a particular wavelength
to some (one or more) of the other nodes on the same wavelength. It means that the central node is able to duplicate a message incoming on a wavelength to different fibres without changing its wavelength.
Therefore if a node $v$ sends a multicast to a set of nodes $S(v)$, $v$ should send the message to the central node on
a set of wavelengths so that the central node redirect it to
each node of $S(v)$ using one of these wavelengths. 
The aim is to minimise the total number of used wavelengths. We refer to Brandt and Gonzalez~\cite{BrGo05} for a more complete description of the model and for some partial results. In what follows, we will briefly explain the main contributions of this paper.

\vspace{11pt}

We first study the basic case when there is a unique fibre between the central node and each node of $V$ and
each vertex $v$ sends a unique multicast $M(v)$ to a set $S(v)$ of
nodes.  In this case, the problem becomes equivalent to {\it directed star colouring}: let $D$ be the digraph with vertex set $V$ such that the
outneighbourhood of a vertex $v$ is $S(v)$. We note that $D$ is a
digraph and not a multidigraph, i.e., there are no parallel arcs in $D$, as $S(v)$
is a set.  The problem is then to find the smallest $k$ such that
there exists a mapping $\phi: A(D)\rightarrow \{1,\dots ,k\}$
satisfying the following two conditions:
\begin{itemize}
\item[$(i)$] For all pair of arcs $uv$ and $vw$, $\phi(uv) \neq \phi(vw)$; 
\item[$(ii)$] For all pair of arcs $uv$ and $u'v$, $\phi(uv) \neq \phi(u'v)$.
\end{itemize}
Such a mapping is called {\it directed star $k$-colouring}.
The {\it directed star arboricity} of a digraph $D$, denoted by $dst(D)$,
is the minimum integer $k$ such that there
exists a directed star $k$-colouring.
This notion has been  introduced by Guiduli in \cite{Gui97} and is an analog of the {\it star arboricity} defined
by Algor and Alon in \cite{AlAl89}.

\vspace{5pt}
 The indegree of a vertex $v$, $d^-(v)$, corresponds to the number of multicasts that $v$ receives.
A sensible assumption on the model is that a node receives a bounded number of multicasts.
Hence, Brandt and Gonzalez~\cite{BrGo05} studied the directed star arboricity
of a digraph $D$ with regards to its maximum indegree.
The {\it maximum indegree} of a digraph $D$, denoted by $\Delta^-(D)$ or 
simply $\Delta^-$ when $D$ is clearly understood from the context, is $\max\:\{\:d^-(v) \tq v\in V(D)\:\}$.
Brandt and Gonzalez showed that 
$dst(D)\leq  \lceil 5\Delta ^-/2\rceil$.
This upper bound is tight if $\Delta^-=1$, because odd circuits 
have directed star arboricity three.
However, as we will show in Section~\ref{entrant}, the upper bound can be improved for larger values of $\Delta^-$.
\begin{theorem}\label{2k+1}
Every digraph $D$ satisfies $dst(D)\leq 2\Delta^-+1$.
\end{theorem}

\noindent We conjecture that

\begin{conjecture}\label{c:central}
Every digraph $D$ with maximum indegree $\Delta^-\geq 2$
satisfies $dst(D)\leq 2\Delta^-$.
\end{conjecture}

\noindent This conjecture would be tight as Brandt~\cite{Bra03} showed that for
every $\Delta^-$, there is an acyclic digraph $D_{\Delta^-}$ with maximum 
indegree $\Delta^-$ and $dst(D_{\Delta^-})=2\Delta^-$. His construction is the special
case for $n=m=1$ of the construction given in
Proposition~\ref{lambdainf}. We settle Conjecture~\ref{c:central} for acyclic
digraphs in Section~\ref{entrant}. So combined with Brand's construction, $2\Delta^-$ is the best bound we can expect for acyclic digraphs.

\begin{remark}\rm
Let us note at this point that we restrict ourselves to
simple digraphs, i.e., we allow circuits of length two but multiple 
arcs are not permitted. When multiple arcs are allowed, all the bounds above do not hold.
Indeed, given an integer $\Delta^-$, the multidigraph $T_{\Delta^-}$ with three vertices 
$u$, $v$ and $w$, and $\Delta^-$ parallel arcs to each of $uv$, $vw$ and $wu$ satisfies $dst(T_{\Delta^-})=3\Delta^-$.
Moreover, this example is extremal since every multidigraph satisfies 
$dst(D)\leq 3\Delta ^-$. This can be shown by induction: pick a vertex $v$ 
with outdegree at most its indegree. (Such a vertex exists since $\sum_{u\in V(D)} d^+(u)=\sum_{u\in V(D)} d^-(u)$.)
If $v$ has no inneighbour, then $v$ is isolated, and we can remove $v$ and apply induction. 
Otherwise, we consider any arc $uv$. The colour of $uv$ must be different from the colours of the $d^-(u)$ arcs entering $u$, the
$d^+(v)$ arcs leaving $v$, and the $d^-(v)-1$ other arcs entering $v$, so at most $3\Delta ^--1$ arcs in total.
Hence, we may remove 
the arc $uv$, apply induction to obtain a colouring of $D \setminus uv$. Extending this colouring to $uv$, we obtain a directed star colouring of $D$ with at most $3\Delta ^-$ colours.
\end{remark}

Note that to prove Conjecture~\ref{c:central}, it will be enough to consider the two cases $\Delta^-=2$ and $\Delta^-=3$.  To see this,
let $D$ be a digraph with maximum indegree $\Delta^-\geq 2$ and $k=\lfloor \Delta^-/2\rfloor$.
For every vertex $v$, let $(N^-_1(v), N^-_2(v), \dots , N^-_k(v))$ be a partition of
$N^-(v)$ such that $|N^-_i(v)|\leq 2$ for all $1\leq i\leq k-1$ and 
$|N^-_k(v)|\leq 2$ if $\Delta^-$ is even and $|N^-_k(v)|\leq 3$ if  $\Delta^-$ is odd. 
Then the digraph $D_i$ with vertex set $V(D)$ and such that $N^-_{D_i}(v)= N^-_{i}(v)$ for every vertex $v\in V(D)$,
has maximum indegree at most two except if $i=k$ and $\Delta^-$ is odd, in which case $D_k$ has maximum indegree
at most three. If Conjecture~\ref{c:central} holds for every $D_i$ then 
it would also hold for $D$.

\vspace{9pt}

We next consider the directed star arboricity of a digraph with bounded maximum
degree. The {\it degree} of a vertex $v$ is $d(v)=d^-(v) + d^+(v)$.  
This corresponds to the degree of
the vertex in the underlying multigraph.  (We have edges with
multiplicity two in the underlying multigraph each time there is a circuit of length two in the digraph.)
The {\it maximum degree} of a digraph $D$, denoted by $\Delta(D)$, or
simply $\Delta$ when $D$ is clearly understood from the context, is
$\max\:\{\:d(v), v\in V(D)\:\}$. 
Let us denote by $\mu(G)$, the maximum multiplicity 
of an edge in a multigraph. 
By Vizing's theorem~\cite{Viz64}, one can colour the edges of a multigraph 
with $\Delta(G) + \mu (G)$
colours so that two edges have different colours if they are incident. 
Since the multigraph underlying a digraph has maximum multiplicity at most two, 
for any digraph $D$, $dst(D)\leq \Delta +2$.
We conjecture the following:
\begin{conjecture}\label{degre}
Let $D$ be a digraph with maximum degree $\Delta \geq 3$.
Then $dst(D)\leq \Delta$.
\end{conjecture}
This conjecture would be tight since every digraph with $\Delta=\Delta^-$
has directed star arboricity at least $\Delta$.
In Section~\ref{seccubic}, we prove that Conjecture \ref{degre} holds when $\Delta=3$.

\begin{theorem}\label{degre3}
Every subcubic digraph has directed star arboricity at most three.
\end{theorem}

A first step towards Conjectures~\ref{c:central} and \ref{degre} 
would be to prove
the following weaker statement. 
\begin{conjecture}\label{diregular}
Let $k\geq 2$ and $D$ be a digraph.
If $\max(\Delta^-, \Delta^+)\leq k$ then $dst(D)\leq 2k$.
\end{conjecture}
This conjecture holds and is far from being tight for large values of $k$.
Indeed Guiduli~\cite{Gui97} showed that if 
$\max(\Delta^-, \Delta^+) \leq k$, then $dst(D)\leq k +20\log k + 84$.
Guiduli's proof is based on the fact that, when both out- and indegrees
are bounded, the colour of an arc depends on the colour of few other
arcs. This bounded dependency allows the use of the Lov\'asz Local
Lemma.  This idea was first used by Algor and Alon~\cite{AlAl89} for
the star arboricity of undirected graphs. We also note that Guiduli's
result is (almost) tight since there are digraphs $D$ with
$\max(\Delta^-, \Delta^+)\leq k$ and $dst(D)\geq k + \Omega (\log k)$
(see \cite{Gui97}).  

\noindent As for Conjecture
\ref{c:central}, it is quite straightforward to check that it is sufficient to prove Conjecture~\ref{diregular}
for $k=2$ and $k=3$.
In Section~\ref{d+d-}, we prove that Conjecture~\ref{diregular} holds
for $k=2$. By the above remark, this implies that
Conjecture~\ref{diregular} holds for all even values of $k$.
\begin{theorem}\label{2diregular}
Let $D$ be a digraph. If $\Delta^-\leq 2$ and $\Delta^+\leq 2$, 
then $dst(D)\leq 4$. In particular, \emph{Conjecture~\ref{diregular}} holds for all even values of $k$.
\end{theorem}

\vspace{11pt}

Next, we study the more general and more realistic problem in which
every vertex of $V$ is connected to the hub by $n$ optical
fibres. Moreover each node may send
several multicasts. We note $M_1(v), \dots ,M_{s(v)}(v)$ the $s(v)$ multicasts that node $v$ sends. For $1 \leq i \leq s(v)$, the set of nodes to which the multicast $M_i(v)$ is sent is denoted by $S_i(v)$. The problem is still to find the minimum number of wavelengths used considering that all fibres are identical.

\noindent We model this as
a problem on {\it labelled} digraphs: We construct a multidigraph $D$ on vertex
set $V$. For each multicast $M_i(v)=(v,S_i(v))$, $v \in V$, $1 \leq i \leq s(v)$, we add the set of arcs
$A_i(v)=\{vw, w\in S_i(v)\}$ with label $i$. The label of an arc $\vec a$
is denoted by $l(\vec a)$. Thus for every ordered pair $(u,v)$ of vertices and
label $i$ there is at most one arc $uv$ labelled by $i$. If each
vertex sends at most $m$ multicasts, there are at most $m$ labels on
the arcs. Such a digraph is said to be {\it $m$-labelled}.  One wishes
to find an {\it $n$-fibre wavelength assignment} of $D$, that is a
mapping $\Phi: A(D)\rightarrow \Lambda \times \{1,\dots , n\} \times
\{1, \dots n\}$ in which every arc $uv$ is associated a triple
$(\lambda (uv),f^+(uv), f^-(uv))$ such that~:
\begin{itemize}
\item[$(i)$] For each pair of arcs $uv$ and $vw$, $(\lambda (uv), f^-(uv)) \neq (\lambda (vw), f^+(vw))$; 
\item[$(ii)$] For each pair of arcs $uv$ and $u'v$, $(\lambda (uv), f^-(uv)) \neq (\lambda (u'v), f^-(u'w))$;
\item[$(iii)$] For each pair of arcs $vw$ and $vw'$, if $l(vw)\neq l(vw')$, then $(\lambda (vw), f^+(vw)) \neq (\lambda (vw'), f^+(vw'))$.
\end{itemize}

\noindent Here $\Lambda$ is the set of available wavelengths, $\lambda (uv)$ corresponds to the wavelength of $uv$,
and $f^+(uv)$ and $f^-(uv)$ are the fibres used in $u$ and $v$, respectively. We can describe the above equations as follows:
\begin{itemize}
\item Condition $(i)$ corresponds to the requirement that an arc entering $v$ and an arc leaving $v$ should have either different wavelengths or different fibres;
\item Condition $(ii)$ corresponds to the requirement that two arcs entering $v$ should have either different wavelengths or different fibres; and finally
\item Condition $(iii)$ corresponds to the requirement that two arcs leaving $v$ with different labels have either different wavelengths or different fibres.

\end{itemize}
The problem is to find the minimum cardinality $\lambda_n(D)$ of $\Lambda$ such that
there exists an  {\it $n$-fibre wavelength assignment} of $D$.

\noindent The crucial part of an $n$-fibre wavelength assignment is the function $\lambda$
which assigns colours (wavelengths) to the arcs.
It must be an {\it $n$-fibre colouring}, that is a function $\phi:  A(D)\rightarrow \Lambda$,
such that at each vertex $v$, for each colour $\omega \in \Lambda$,
$in(v,\omega)+out(v,\omega)\leq n$ where $in(v,\omega)$ denotes the number of arcs coloured by $\omega$ entering $v$ and
$out(v,\omega)$ denotes the number of labels $l$ such that there exists an arc leaving $v$ coloured by $\omega$.
Once we have an $n$-fibre colouring, one can easily find a suitable wavelength assignment. For every vertex $v$ and every colour $\omega$, this is done by assigning
a different fibre to each arc of colour $\omega$ entering $v$, and to each set of arcs of colour $\omega$ of the same label
that leave $v$.
We conclude that  $\lambda_n(D)$ is the minimum number of colours such that there exists an $n$-fibre colouring.

\noindent We are particularly interested in $\lambda_n(m,k)=\max\{\lambda_n(D) \tq
D \mbox{ is $m$-labelled} \et \Delta^-(D)\leq k\}$, that is the maximum
number of wavelengths that may be necessary if there are $n$ fibres,
and each node sends at most $m$ multicasts and receives at most $k$ multicasts.
In particular, $\lambda_1(1,k)=\max\{dst(D) \tq \Delta^-(D)\leq k\}$.
(So our above mentioned results show that $2k\leq \lambda_1(1,k)\leq 2k+1$.)
Brandt and Gonzalez showed that for $n\geq 2$ we have
$\lambda_n(1,k)\leq \left\lceil\frac{k}{n-1} \right\rceil$.
In Section~\ref{+rsfibres}, we study the case when $n\geq 2$ and $m\geq 2$.
We show in Proposition~\ref{lambdainf} and Theorem~\ref{fibmulti} that  $$ \mbox{if $m\geq n$ then}\ \ \ 
 \left\lceil\frac{m}{n}\left\lceil \frac{k}{n}\right\rceil +   \frac{k}{n} \right\rceil\: \leq  \: \lambda_n(m,k) 
\: \leq \: \left\lceil\frac{m}{n}\left\lceil \frac{k}{n}\right\rceil +   
\frac{k}{n} \right\rceil + C \frac{m^2\log k}{n} \mbox {\ \ \  for some constant $C$.}$$
We conjecture that the lower bound is the right value of $\lambda_n(m,k)$ when $m\geq n$.
We also show in Proposition~\ref{lambdainf} and Proposition~\ref{upper} that  
$$\mbox {if $m<n$, then \hspace{.5cm}} \left\lceil\frac{m}{n}\left\lceil \frac{k}{n}\right\rceil +   
\frac{k}{n} \right\rceil\leq  \lambda_n(m,k) \leq \left\lceil\frac{k}{n-m}\right\rceil.$$
The lower bound generalises Brandt and Gonzalez~\cite{BrGo05} results which established this inequality in the particular cases when
$k\leq 2$, $m\leq 2$ and $k=m$.  The digraphs used to show this lower
bound are all acyclic. We show that if $m \geq n$ then this lower bound is tight for acyclic digraphs.
Moreover the above mentioned digraphs have large outdegree. Generalising the result of Guiduli~\cite{Gui97},
we show that for an $m$-labelled digraph $D$ with both in- and outdegree  bounded by $k$ only few colours are needed when $m\geq n$:
$$ \lambda_n(D)\leq \frac{k}{n} + C'  \frac{m^2\log k}{n} \mbox {\ \ \ \ \ for some constant $C'$.}$$

\vspace{11pt}

Finally, in Section~\ref{sec:conclusion}, we consider the complexity of finding the directed star arboricity of a digraph, and prove that, unsurprisingly, this is an ${\cal NP}$-hard problem. More precisely, we show that determining the directed star arboricity of a digraph with  in- and outdegree
at most two is ${\cal NP}$-complete. We then give a very short proof of a theorem of Pinlou and Sopena~\cite{PiSo04}, showing that  {\it acircuitic directed star arboricity of subcubic graphs is at most four} (see Section~\ref{sec:conclusion} for the definitions).

\section{Directed Star Arboricity of Digraphs with Bounded Indegrees}\label{entrant}
 In this section, we give the proof of Theorem~\ref{2k+1} and settle Conjecture~\ref{c:central} for acyclic digraphs.

An {\it arborescence} is a connected digraph in which every vertex has indegree one
except one, called {\it root}, which has indegree zero.
A {\it forest} is the disjoint union of arborescences.
A {\it star} is an arborescence in which the root dominates all the
other vertices. A {\it galaxy} is a forest of stars. 
Clearly, every colour class of a directed star colouring is a galaxy.
Hence, the directed star arboricity of a digraph $D$ is the minimum
number of galaxies into which $A(D)$ may be partitioned. 

It is easy to see that a forest has directed star arboricity at most two.
Hence, an idea to prove Conjecture~\ref{c:central} would be to show that
every digraph has an arc-partition into $\Delta^-$ forests.
However this statement is false. Indeed a theorem of Frank~\cite{Fra79} (see also Chapter 53 of~\cite{Sch03}) characterises all digraphs which have an arc-partition into
$k$ forests.
Let $D=(V,A)$.
For any $U\subset V$, the digraph induced by the vertices of $U$ is denoted
$D[U]$.
\begin{theorem}[A. Frank]\label{CNSarbo}
A digraph $D=(V,A)$ has an arc-partition into $k$ forests if and only if 
$\Delta^-(D)\leq k$ and for every $U\subset V$, the digraph $D[U]$ has at most $k(|U|-1)$ arcs.
\end{theorem}
 This theorem implies that every digraph $D$ has an arc-partition into $\Delta^-+1$
forests. Indeed for any $U\subset V$, $\Delta^-(D[U])\leq \min\{\Delta^-, |U|-1\}$,
so  $D[U]$ has at most $\min\{\Delta^-, |U|-1\}\times |U|\leq (\Delta^-+1)(|U|-1)$ arcs.
Hence, every  digraph has directed star arboricity at most $2\Delta^-+2$.

\begin{corollary}\label{2k+2}
Every digraph $D$ satisfies $dst(D)\leq 2\Delta^-+2$.
\end{corollary}

 Theorem~\ref{2k+1} states that $dst(D)\leq 2\Delta^-+1$. 
The idea to prove this theorem is to show that  
every digraph has an arc-partition into $\Delta^-$ forests 
and a galaxy $G$. To do so, we prove a stronger result, Lemma~\ref{decomp} below.

We need some extra definitions.
A {\it sink} is a vertex with outdegree $0$. A {\it source} is 
a vertex with indegree $0$. A multidigraph $D$ will be called {\it $k$-nice} if 
$\Delta^-\leq k$, and if the tails of parallel arcs, if any, are sources.
A {\it $k$-decomposition} of $D$ is an arc-partition into $k$ forests 
and a galaxy $G$ such that 
every source of $D$ is isolated in $G$.
Let $u$ be a vertex of $D$. A $k$-decomposition of $D$ is 
{\it $u$-suitable} if no arc of $G$ has head $u$.

\begin{lemma}\label{decomp}
Let $u$ be a vertex of a $k$-nice multidigraph $D$.
Then $D$ has a $u$-suitable $k$-decomposition.
\end{lemma}

\begin{proof}
We proceed by induction on $n+k$ by considering (strong) connectivity
of $D$:
\begin{itemize}
\item If $D$ is not connected as graph, we apply induction on every component.
\item If $D$ is strongly connected, every vertex has indegree at least
one. (Recall that there are no parallel arcs.) 
Let $v$ be an outneighbour of $u$. There
exists a spanning arborescence $T$ with root $v$ which contains all the arcs
with tail $v$. Let $D'$ be the digraph obtained from $D$ by removing the arcs 
of $T$ and $v$. Observe that $D'$ is $(k-1)$-nice. 
By induction, it has a $u$-suitable $(k-1)$-decomposition
$(F_1, \dots , F_{k-1}, G)$. Note that each $F_i$, for $1\leq i\leq k-1$, $T$ and 
$G$ contain all the arcs of $D$ except those with head $v$.
By construction, $G'=G\cup uv$ is a galaxy since no arc
of $G$ has head $u$. Let $u_1, \dots ,u_{l-1}$ be the inneighbours of $v$ 
distinct from $u$, where $l\leq k$. Let $F'_i=F_i\cup u_iv$,
for all $1\leq i\leq l-1$.
Each $F'_i$ is a forest, so $(F_1, \dots , F_{k-1}, T, G')$ is 
a $u$-suitable $k$-decomposition of $D$.
\item In the only remaining case, $D$ is connected but not strongly connected. We consider a terminal strongly 
connected component $D_1$ of $D$. Set $D_2=D\setminus D_1$.
Let $u_1$ and $u_2$ be two vertices of $D_1$ and $D_2$, respectively, 
such that $u$ is one of them. 

If $D_2$ has a unique vertex $v$ 
(thus $u_2=v$), since $D$ is connected and $D_1$ is strong, there exists a spanning arborescence $T$ of $D$
with root $v$.  Now $D'=D\setminus A(T)$ is a
$(k-1)$-nice multidigraph, so by induction it has a $u$-suitable
$(k-1)$-decomposition. Adding $T$ to this decomposition, we 
obtain a $u$-suitable $k$-decomposition.
If $D_2$ has more than one vertex, it admits  a $u_2$-suitable 
$k$-decomposition $(F_1^2, \dots , F_k^2, G^2)$, by induction.
Moreover the digraph $D'_1$ obtained by contracting $D_2$ to a single vertex $v$ is $k$-nice and so 
has a $u_1$-suitable $k$-decomposition $(F_1^1, \dots , F_k^1, G^1)$.
Moreover, since $v$ is a source, it is isolated in $G^1$.
Hence $G=G^1\cup G^2$ is a galaxy.
We now let $F_i$ be the union of $F_i^1$ and $F_i^2$ by replacing
the arcs of $F^1_i$ with tail $v$ by the corresponding arcs in $D$.
Then $(F_1, \dots , F_k, G)$ is a $k$-decomposition of $D$ which is suitable
for both $u_1$ and $u_2$.
\end{itemize}
\vspace{-.5cm}
\end{proof}

Theorem~\ref{2k+1} is an immediate consequence of Lemma~\ref{decomp}.

\subsection{Acyclic Digraphs}

It is not very hard to show that $dst(D)\leq 2\Delta^-$ when $D$
is acyclic, but we will prove this result in a more constrained
way. For $n\leq p$, a {\it cyclic $n$-interval} of $\{1,2, \dots , p\}$ is a set
of $n$ consecutive numbers modulo $p$.
Now for the directed star colouring, we will insist that for 
every vertex $v$, the (distinct)
colours used to colour the arcs with head $v$ are chosen in a
{\it cyclic $k$-interval} of $\{1,2, \dots , 2k\}$.
Thus, the number of possible sets of colours used to colour 
the entering arcs of a vertex $v$ drastically falls from
$2k\choose d^-(v)$ when every set is {\it a priori} possible,
to at most $2k\times {k\choose d^-(v)}$. Note that having consecutive colours on the arcs entering a vertex
corresponds to having consecutive wavelengths on the link between the corresponding
node and the central one. 
This may of importance for issues related to grooming in optical networks. For details
about grooming, we refer the reader to the two comprehensive surveys~\cite{LGP01,MoLi01}.

\begin{theorem}\label{dst-acircuitic}
Let $D$ be an acyclic digraph with maximum indegree $k$. Then $D$ admits a directed star $2k$-colouring such that for every vertex,
the colours assigned to its entering arcs are included in a cyclic 
$k$-interval of  $\{1,2, \dots , 2k\}$.
\end{theorem}

To prove this theorem, we first state and prove the following result on sets of distinct representatives.

\begin{lemma}\label{representatives}
Let $I_1, \dots , I_k$ be $k$ non necessarily distinct cyclic $k$-intervals of $\{1,2, \dots , 2k\}$.
Then $I_1, \dots , I_k$ admit a set of distinct representatives forming
a cyclic $k$-interval.
\end{lemma} 
\begin{proof}
We consider $I_1, \dots , I_k$ as a set of $p$
{\it distinct} cyclic $k$-intervals $I_1, \dots , I_p$ with respective
multiplicity $m_1, \dots , m_p$ such that $\sum_{i=1}^p m_i = k$. Such a system will be denoted by $((I_1,m_1),\dots,(I_p,m_p))$.
 We shall prove the existence of a cyclic $k$-interval $J$, such that $J$ can be partitioned into $p$ subsets $J_i$, $1\leq i\leq p$,
such that $|J_i|=m_i$ and $J_i\subset I_i$. This proves the lemma (by associating distinct elements of $J_i$ to each copy of $I_i$).

We proceed by induction on $p$. The result holds trivially for $p=1$. We have to deal with two cases:
\begin{itemize} 
\item There exist $i$ and  $j$ such that
$|I_j\setminus I_i|=|I_i\setminus I_j| \leq \max(m_i,m_j)$. 

Suppose without loss of generality that $i<j$ and $m_i\geq m_j$. We
apply the induction hypothesis to $((I_1,m_1),
\cdots,(I_i,m_i+m_j), \cdots ,(I_{j-1},m_{j-1}),(I_{j+1},m_{j+1}),\cdots, (I_p,m_p))$,
in order to find a cyclic interval $J'$, such that $J'$ admits a
partition into subsets $J'_r$, such that for any $r$ different from  $i$ and $j$, the set $J'_r
\subset I_r$ is a subset of size $m_r$, and $J'_i\subset I_i$ is of
size $m_i+m_j$. We now partition $J'_i$ into two sets $J_i$ and $J_j$
with respective size $m_i$ and $m_j$, in such a way that
$(I_i\setminus I_j)\cap J'_i\subseteq J_i$. Remark that this is
possible precisely because of our assumption $|I_j\setminus
I_i|=|I_i\setminus I_j| \leq m_i$. Since $J_i\subset I_i$ and
$J_j\subset I_j$, this refined partition of $J'$ is the desired one.

\item For any $i,j$ we have $|I_j\setminus I_i|=|I_i\setminus I_j| \geq \max(m_i,m_j)+1$. 

Each $I_i$ intersects exactly $2m_i -1$ other cyclic $k$-intervals on
less than $m_i$ elements. Since there are $2k$ cyclic $k$-intervals in
total and $\sum_{i=1}^p (2m_i -1) = 2k-p < 2k$, we conclude the
existence of a cyclic $k$-interval $J$ which intersects each $I_i$ in
an interval of size at least $m_i$.

Let us prove that one can partition $J$ in the desired way.  By Hall's
matching theorem, it suffices to prove that for every subset ${\cal
I}$ of $\{1, \dots ,p\}$, we have $|\bigcup_{i\in {\cal I}} I_i\cap J| \geq
\sum_{i\in {\cal I}} m_i$.

Suppose for the sake of a contradiction that a subset ${\cal I}$ of $\{1, \dots
,p\}$ violates this inequality. Such a subset will be called {\it
contracting}. Without loss of generality, we assume that ${\cal I}$ is
a contracting set with minimum cardinality and that ${\cal
I}=\{1,\dots, q\}$. Observe that by the choice of $J$, we have $q \geq
2$. The set $K:=\bigcup_{i\in {\cal I}} I_i\cap J$ consists of one or
two intervals of $J$, each containing one extremity of $J$. By the
minimality of ${\cal I}$, $K$ must be a single interval (if not, one
would take ${\cal I}_1$ (resp. ${\cal I}_2$), all the elements of
${\cal I}$ which contains the first (resp. the second) extremity of
$J$. Then one of ${\cal I}_1$ or ${\cal I}_2$ would be contracting).  Thus, one of
the two extremities of $J$ is in every $I_i$, $i\in {\cal I}$. Without
loss of generality, we may assume that $(I_1\cap J) \subset (I_2\cap
J)
\subset \dots \subset (I_q\cap J)$.
Now, for every $2\leq i\leq q$, $|I_{i}\setminus I_{i-1}| =|(I_{i}\cap J)
\setminus (I_{i-1}\cap J)|\geq  \max(m_i,m_{i-1})+1\geq m_i+1$.
But $|\bigcup_{i\in {\cal I}} I_i\cap J|= |(I_1\cap J)| + \sum_{i=2}^q 
|(I_{i}\cap J) \setminus (I_{i-1}\cap J)|$.
So  $|\bigcup_{i\in {\cal I}} I_i\cap J|\geq \sum_{i=1}^q m_i +q-1$,
which is a contradiction.
\end{itemize}
\vspace{-.5cm}
\end{proof}

\par \noindent {\bf Proof of Theorem~\ref{dst-acircuitic}}. 
By induction on the number of vertices, the result being trivial if
$D$ has one vertex.
Suppose now that $D$ has at least two vertices. Then $D$ has a sink $x$.
By the induction hypothesis, $D\setminus x$ has a directed star $2k$-colouring 
$c$ such that for every vertex, the colours assigned to its entering 
arcs are included in a cyclic $k$-interval.
Let $v_1, v_2, \dots ,v_l$ be the inneighbours of $x$ in $D$, where $l\leq k$
because $\Delta^-(D)\leq k$.
For each $1\leq i\leq l$, let $I'_i$ be a cyclic $k$-interval which contains
all the colours of the arcs with head $v_i$. We set 
$I_i=\{1,\dots ,2k\}\setminus I'_i$. Clearly, $I_i$ is a 
cyclic $k$-interval and the arc $v_ix$ can be coloured by any 
element of $I_i$. By Lemma~\ref{representatives},
$I_1, \dots , I_l$ have a set of distinct representatives 
included in a cyclic $2k$-interval $J$. Hence assigning $J$ to $x$, and colouring the arc $v_ix$ by 
the representative of $I_i$ gives a directed star $2k$-colouring of $D$.
\hfill$\Box$

\vs

Theorem~\ref{dst-acircuitic} is tight : Brandt~\cite{Bra03} showed
that for every $k$, there is an acyclic digraph such that
$\Delta^-(D_k)=k$ and $dst(D_k)=2k$.  His construction is the special
case for $n=m=1$ of the construction given in
Proposition~\ref{lambdainf}.

\section{Directed Star Arboricity of Subcubic Digraphs}\label{seccubic}

Recall that a {\it subcubic} digraph is a graph with degree at most three.
In this section, we give the proof of Theorem~\ref{degre3} which states that the directed star arboricity of a
subcubic digraph is at most three. 

To do so, we need to establish some preliminary lemmas which will enable us to extend a
partial directed star colouring into a directed star colouring of the
whole digraph.  To state these lemmas, we need the following definition.  Let
$D=(V,A)$ be a digraph and $S$ be a subset of $V\cup A$.  Suppose that
each element $x$ of $S$ is assigned a list $L(x)$.  A colouring $c$ of
$S$ is an {\it $L$-colouring} if $c(x)\in L(x)$ for every $x\in S$.

\begin{lemma}\label{cycle}
Let $C$ be a circuit in which every vertex $v$ receives a list $L(v)$ 
of two colours among $\{1,2,3\}$ and each arc $\vec a$ receives the
list $L(\vec a)=\{1,2,3\}$. The following two statements are equivalent:
\begin{itemize}
\item There is no $L$-colouring $c$ of the arcs and vertices 
such that  $c(x)\neq c(xy)$, $c(y)\neq c(xy)$, and $c(xy)\neq c(yz)$, for all arcs
$xy$ and $yz$.
\item $C$ is an odd circuit and all the vertices have the same list.
\end{itemize}
\end{lemma}

\begin{proof}
Assume first that every vertex is assigned the same list, say $\{1,2\}$. If $C$
is odd, it is a simple matter to check that we can not find the desired colouring.
Indeed, among two consecutive arcs, one has to be coloured $3$.
If $C$ is even, we
colour the vertices by $1$ and  the arcs  alternately by $2$ and $3$.

Now assume that $C=x_1x_2\dots x_kx_1$ and $x_1$ and $x_2$ are assigned
different lists. Say $L(x_1)=\{1,2\}$ and $L(x_2)=\{2,3\}$. We colour
the arc $x_1x_2$ by $3$, the vertex $x_2$ by $2$ and the arc $x_2x_3$ by $1$. 
Then we colour $x_3$, $x_3x_4$, \dots , $x_k$ greedily. 
It remains to colour $x_kx_1$ and $x_1$.  Two cases may happen:
If we can colour $x_kx_1$ by 1 or 2, we do it and colour $x_1$ by $2$ or $1$ respectively.
Otherwise the set of colours  assigned to $x_k$ and $x_{k-1}x_k$ is $\{1,2\}$. 
Hence, we colour $x_kx_1$ with $3$, $x_1$ by $1$, and recolour $x_1x_2$ by $2$ and $x_2$ by $3$.
\end{proof}

\begin{lemma}\label{extension}
Let $D$ be a subcubic digraph with no vertex of outdegree two and
indegree one. Suppose that every arc $\vec a$ has a list of colours $L(\vec a) \subset \{1,2,3\}$ such that:

\begin{itemize}

\item If the head of $\vec a$ is a sink $s$ (in which case, $\vec a$ will be called a {\it final arc}), $|L(\vec a)|\geq d^-(s)$.

\item If $\vec a$ is not a final arc and the tail of $\vec a$ is a source 
(in which case, $\vec a$ will be called an {\it initial arc}), $|L(\vec a)|\geq 2$.

\item In all the other cases, $|L(\vec a)|= 3$.
\end{itemize}
In addition, assume that the followings hold:
\begin{itemize}
\item If a vertex is the head of at least two initial arcs $\vec a$ and $\vec b$, the union of the
lists of colours $L(\vec a)$ and $L(\vec b)$ contains all the three colours.

\item If all the vertices of an odd circuit are the tails of initial arcs, 
the union of  the lists of colours of these initial arcs contains all the three
colours.
\end{itemize}

Then $D$ has a directed star $L$-colouring. 
\end{lemma}

\begin{proof}
We colour the graph inductively. Consider a terminal strong component $C$
of $D$. Since $D$ has no vertex with indegree one and outdegree two,
$C$ induces either a singleton or a circuit.

\begin{itemize}

\item[1)] Assume that $C$ is a singleton $v$ which is the head of a unique arc $\vec a=uv$.  
If $u$ has indegree zero, we colour $\vec a$ with a colour of its list. If $u$ has 
indegree one, and thus total degree two, we colour $\vec a$ by the colour of its list
and remove this colour from the list of the arc with head $u$.
If $u$ is the head of $\vec e$ and $\vec f$, observe that $L(\vec e)$ and $L(\vec f)$ 
have at least two colours and their union have all the three colours.
To conclude, we colour $\vec a$ with a colour in its list, remove this colour from
$L(\vec e)$ and $L(\vec f)$, remove $\vec a$, and split $u$ into two vertices, one with head
$\vec e$ and the other with head $\vec f$. Now, we choose different colours for the arcs $\vec e$ and $\vec f$ in their respective lists to form the new list $L(\vec e)$
and $L(\vec f)$.

\item[2)] Assume that $C$ is a singleton $v$ which is the head of several arcs, including
$\vec a=uv$. In this case, we reduce $L(\vec a)$ to a single colour, remove this
colour from the other arcs with head $v$ and split $v$ into $v_1$,
which becomes the head of $\vec a$, and $v_2$ which becomes the head of the other 
arcs.

\item[3)] Assume that $C$ is a circuit. Every arc entering $C$ has 
a list of at least two colours. We can apply Lemma~\ref{cycle} to conclude.

\end{itemize}
\end{proof}

\par \noindent {\bf Proof of Theorem~\ref{degre3}}.
Assume for the sake of a contradiction that the digraph $D$ has directed star arboricity
more than three and is minimum for this property with respect to the number of arcs.
Observe that $D$ has no source, otherwise we simply delete
it with all its incident arcs, apply induction and extend the colouring. This is possible
since arcs leaving from a source can be coloured arbitrarily.
Let $D_1$ be the subdigraph of $D$ induced by the vertices 
of indegree at most $1$. We denote by $D_2$ the digraph induced by the other 
vertices, and by $[D_i,D_j]$ the set of arcs with tail in $D_i$ and head in $D_j$.
We claim that $D_1$ contains no even circuit. If not, we simply remove the arcs
of this even circuit, apply induction. We can extend the colouring to the arcs of 
the even circuit since every arc of the circuit has two colours available.

A {\it critical set} of vertices of $D_2$ is either a vertex of $D_2$ with indegree
at least two in $D_1$, or an odd circuit of $D_2$ having all its inneighbours
in $D_1$. Observe that critical sets are disjoint. For every critical
set $S$, we select two arcs entering $S$ from $D_1$, called {\it selected arcs} of $S$.

Let $D'$ be digraph induced by the arc set $A'=A(D_1)\cup [D_2,D_1]$.
We now define a {\it conflict graph} on the arcs of $D'$ in the following way:

\begin{itemize}
\item Two arcs $xy$, $yv$ of $D'$ are in conflict, called {\it normal conflict} at $y$.

\item Two arcs $xy$, $uv$ of $D'$ are also in conflict if there exists 
two selected arcs of the same set $S$ with tails $y$ and $v$. These conflicts
are called {\it selected conflicts} at $y$ and $v$.
\end{itemize}
Let us analyse the structure of the conflict graph. 
Observe first that an arc is in conflict with three arcs : one normal conflict at its tail 
and at most two (normal or selected) at its head. We claim that there is no $K_4$ in the conflict graph. 
For the sake of a contradiction, suppose there is one. This means that there are four arcs $\vec a,\vec b,\vec c$ and $\vec d$ pairwise in conflict.
Since each of these arcs have degree three in $K_4$, each of these arcs should have a normal conflict at its tail, 
and so the digraph induced by these four arcs contains a circuit.
This circuit cannot be of even length (two or four) so it has to be of length three.
It follows that the four arcs $\vec a,\vec b,\vec c$ and $\vec d$ are as in Figure~\ref{complet} below (modulo a permutation of the labels).
Let $D^*$ be the digraph obtained from $D$ by removing the arcs 
$\vec a,\vec b,\vec c,\vec d$ and their four 
incident vertices. By minimality of $D$, $D^*$ admits a directed star $3$-colouring which
can be extended to $D$ as depicted
below depending if the two leaving arcs are coloured the same or differently. 
This proves the claim.

\begin{figure}[!hbt]
\begin{center}
\epsfig{file=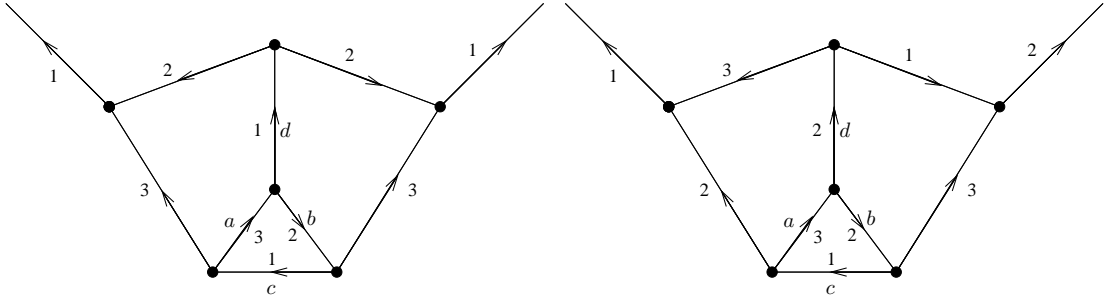, height=4cm}
\caption{A $K_4$ in the conflict graph and the two ways of extending the colouring.}\label{complet}
 \end{center}
\end{figure}

Brooks Theorem asserts that every subcubic graph
without $K_4$ is 3-colourable. So the 
conflict graph admits a $3$-colouring $c$. 
This gives a colouring of the arcs of $D'$. Let $D''$ be the 
digraph obtained from $D$, and let $L$ be the list-assignment on the arcs of $D''$
defined simultaneously as follow:
\begin{itemize}
\item Remove the arcs of $D_1$ from $D$,
\item Assign to each arc of $[D_2,D_1]$ the singleton list containing the colour it has in $D'$, 
\item For each arc $uv$ of $[D_1,D_2]$, there is a unique arc $tu$ in $A(D')$. Assign to $uv$ the list $L(uv)=\{1,2,3\}\setminus c(tu)$.
\item Assign the list $\{1,2,3\}$ to the other arcs.
\item If there are vertices with indegree one and outdegree
two (they were in $D_1$), split each of them into one source of degree two and a sink
of degree one. 
\end{itemize}

Note that there is a trivial one-to-one correspondence between $A(D'')$ and $A(D)\setminus A(D')$.
By the definition of the conflict graph and $D''$, one can easily check that $D''$ and $L$ 
satisfies the condition of Lemma~\ref{extension}.
Hence $D''$ admits a directed star $L$-colouring which union with $c$ is
a directed star $3$-colouring of $D$, a contradiction. The proof of Theorem~\ref{degre3} is now complete.
\hfill$\Box$

\section{Directed Star Arboricity of Digraphs with Maximum In- and 
Outdegree Two}\label{d+d-}

The goal of this section is to prove Theorem~\ref{2diregular}: {\it Every digraph 
with outdegree and indegree at most two has directed 
star arboricity at most four.} 
However, the class of digraphs with in- and outdegree at most
two is certainly not an easy class with respect to directed 
star arboricity, as we will show in Section~\ref{complexity}.

\vspace{11pt}

In order to prove Theorem~\ref{2diregular}, it suffices to show that
$D$ contains a galaxy $G$ which spans 
all the vertices of degree four. Indeed, if this is true, then $D'=D-A(G)$ has maximum 
degree at most $3$ and  by Theorem~\ref{degre3}, $dst(D')\leq 3$. So $dst(D)\leq 4$. Hence Theorem~\ref{2diregular} is directly implied by the following lemma:

\begin{lemma}\label{spanninggalaxy}
Let $D$ be a digraph with maximum indegree and outdegree two.
Then $D$ contains a galaxy which spans the set of vertices with
degree four.
\end{lemma}

To prove this lemma, we need some preliminaries.\\
Let $V$ be a set. An {\it ordered digraph} on $V$ 
is a pair $(\leq,D)$ where:

\begin{itemize} 

\item $\leq$ is a partial order on $V$;

\item  $D$ is a digraph with vertex set $V$;

\item $D$ contains the Hasse diagram of $\leq$. I.e., 
when $x\leq y\leq z$ implies $x=y$ or $y=z$, then $xz$
is an arc of $D$;

\item If $xy$ is an arc of $D$, the vertices $x,y$ are $\leq$-comparable.

\end{itemize}
 
The arcs $xy$ of $D$ thus belong to two different types:
the {\it forward arcs}, when $x\leq y$, and the  {\it backward arcs}, when 
$y\leq x$.

\begin{lemma} \label{ordig}
Let $(\leq ,D)$ be an ordered digraph on $V$.
Assume that every vertex 
is the tail of at most one backward arc and at most two forward arcs,
and that the indegree of every vertex of $D$ is
at least two, except possibly one vertex $x$ with indegree one.
Then $D$ contains two arcs  $\Cc\Aa$ and $\Bb\Dd$ such that $\Aa\leq \Bb\leq \Cc$,
$\Bb\leq \Dd$ and  $\Cc\not \leq \Dd$, all four vertices being distinct except
possibly  $\Aa=\Bb$.

\end{lemma}

\begin{proof}
For the sake of a contradiction, let us consider a counterexample with minimum $|V|$.

An {\it interval } is a subset $I$ of $V$ which has a
minimum $m$ and a maximum $M$ such that $I=\{z:m\leq z\leq M\}$.
An interval $I$ is {\it good} if every arc with tail in $I$ and head 
outside $I$ has tail $M$ and every backward arc in $I$ has tail $M$.

Let $I$ be an interval of $D$.
The digraph $D/I$ obtained from $D$ by {\it contracting $I$}
is the digraph with vertex set $(V\setminus I) \cup \{v_I\}$ such that
$xy$ is an arc if and only either $v_I\notin \{x,y\}$ and $xy\in A(D)$,
or $x=v_I$ and there exists $x_I\in I$ such that $x_Iy\in A(D)$, or
$y=v_I$ and there exists $y_I\in I$ such that $xy_I\in A(D)$.

Similarly, the binary relation 
$\leq _{/I}$ obtained from $\leq$ by {\it contracting $I$} is
the binary relation  on  $(V\setminus I) \cup \{v_I\}$
such that $x\leq_{/I} y$ if and only if either $v_I\notin \{x,y\}$ and $x\leq y$,
or $x=v_I$ and there exists $x_I\in I$ such that $x_I\leq y$, or
$y=v_I$ and there exists $y_I\in I$ such that $x\leq y_I$.
We claim that if $I$ is good then $\leq _{/I}$ is a partial order.
Indeed suppose it is not, then there are two elements $u$ and $t$ such that
$u\leq_{/I} v_I$, $v_I\leq t$ and $u\not\leq_{/I} t$.
Then $M\not\leq t$. Let $\Aa\in I$ be the maximal element of $I$ 
such that $\Aa\leq t$, $\Dd$ be a successor of $\Aa$ in $I$, and
$\Cc$ a successor of $\Aa$ not in $I$ (it exists as $t\notin I$ and $\Aa\leq t$
and maximal in $I$ with this property). Then $\Dd$ and $\Cc$ are incomparable, and $\Aa\Cc$ and $\Aa\Dd$ are in the Hasse diagram
of $\leq$. Because $I$ is good, it follows that $\Cc\Aa$ and $
\Aa\Dd$ are arcs of $D$, which is impossible as $D$ is supposed to be a counterexample.

Hence, if $I$ is a good interval, then  $(\leq_{/I}, D/I)$ is an ordered digraph.
Note that if $x\leq_{/I} v_I$, then $x\leq M$ with $M$ the maximum  of $I$.
The crucial point is that if $I$ a good interval of $D$ for which the conclusion
of Lemma~\ref{ordig} holds for $(\leq_{/I}, D/I)$, 
then it holds for $(\leq, D)$. 
Indeed, suppose there exists two arcs  $\Cc\Aa$ and $\Bb\Dd$ of $D/I$ such that $\Aa\leq_{/I}  \Bb\leq_{/I}  \Cc$,
$\Bb\leq_{/I}  \Dd$, and  $\Cc\not \leq_{/I}  \Dd$. 
Note that since $I$ is good, we have $v_I\neq \Cc$. Let $M$ be the maximum of $I$.\\
If $v_I\notin \{\Aa,\Bb,\Cc,\Dd\}$, then $\Cc\Aa$ and $\Bb\Dd$ gives the 
conclusion for $D$.\\
If $v_I=\Aa$, then $\Cc M$ is an arc. Let us show that $M\leq \Bb$.
Indeed, let $x$ be a maximal vertex in $I$ such that $x\leq \Bb$ and let $y$ be a minimal vertex
such that $x\leq y\leq \Bb$.
Since the Hasse diagram of $\leq $ is included in $D$, $xy$ is an arc, and so $x=M$ (since $I$ is good).
Thus $\Cc M$ and $\Bb\Dd$ are the desired arcs.\\
If $v_I=\Bb$, then $M\Dd$ is an arc and $\Aa\leq M$, so $\Cc\Aa$ and $M\Dd$ are the desired arcs. \\
If $v_I=\Dd$, then there exists $\Dd_I\in I$ such that $\Bb\Dd_I$, so $\Cc\Aa$ and $\Bb\Dd_I$ are the desired arcs. 

Hence to get a contradiction, it is sufficient to find
a good interval $I$ such that  $(\leq_{/I}, D/I)$ satisfies the hypotheses of Lemma~\ref{ordig}.

Observe that there are at least two backward arcs. Indeed, if there
are two minimal elements for $\leq$, there are at least three backward
arcs entering these vertices (since one of them can be $x$). And if there 
is a unique minimum $m$, by letting $m'$ minimal in $V\setminus m$, at least
two arcs are entering $m,m'$.

Let $M$ be a vertex  which is the
tail of a backward arc and which is minimal for $\leq$ for this property.
Since two arcs cannot have the same tail, $M$ is not the maximum of $\leq$ (if any). 
Let $Mm$ be the backward arc with tail $M$. 

We claim that the interval
$J$ with minimum $m$ and maximum $M$ is good. Indeed, by the definition
of $M$, no backward arc has its tail in $J\setminus \{M\}$. Moreover, any forward 
arc $\Bb\Dd$ with its tail in $J\setminus \{M\}$ and its head outside $J$
would give our conclusion (with $\Aa=m$ and $\Cc=M$), a contradiction.

Now consider a good interval $I$ with maximum $M$ which
is maximal with respect to inclusion. We claim that if $x\in I$,
then there is at least one arc entering $I$, and if $x\notin I$,
there are at least two arcs entering $I$ with different tails.

Call $m_1$ the minimum of $I$ and $m_2$ any minimal element of $I\setminus m_1$.
First assume that $x$ is in $I$. There are at least three arcs with heads
$m_1$ or $m_2$. One of them is $m_1m_2$, one of them can be with tail $M$,
but there is still one left with tail not in $I$.
Now assume that $x$ is not in $I$. There are at least two arcs with heads
$m_1$ or $m_2$ and tails not in $I$. If the tails are different,
we are done. If the tails are the same, say $v$, observe that $vm_1$
and $vm_2$ are both backward or both forward (otherwise $v$ would be 
in $I$). Since both cannot be backward, both $vm_1$
and $vm_2$ are forward. 
Hence the interval with minimum $v$ and maximum
$M$ is a good interval, contradicting the maximality of $I$.
This proves the claim.

\noindent This in turn implies that $(\leq_{/I}, D/I)$ satisfies the hypotheses of Lemma~\ref{ordig},
yielding a contradiction.
\end{proof}

\noindent {\bf Proof of Lemma~\ref{spanninggalaxy}.\ }
Let $G$ be a galaxy of $D$ which spans
a maximum number of vertices of degree four. 
Suppose for the sake of a contradiction that some vertex $x$ with degree four 
is not spanned. 

An {\it alternating path} is an oriented path ending at $x$, 
starting by an arc of $G$, and 
alternating with arcs of $G$ and arcs of $A(D)\setminus A(G)$.
We denote by $\cal A$ the set of arcs of $G$ which belong to
an alternating path.

\begin{claim}
Every arc of $\cal A$ is a component of $G$.
\end{claim}

\begin{proof} Indeed, if $uv$ belongs to $\cal A$, it starts some 
alternating path $P$. Thus, if $u$ has outdegree
more than one in $G$, the digraph with the set of arcs 
$A(G)\triangle A(P)$ is a galaxy and spans $V(G)\cup x$. 
\end{proof}

\begin{claim}\label{cl:nocircuit}
There is no circuits alternating arcs of $\cal A$ and
arcs of $A(D)\setminus {\cal A}$.
\end{claim}

\begin{proof} Assume that there is such a circuit $C$. Consider a shortest
alternating path $P$ starting with some arc of $\cal A$ in $C$.
Now the digraph with arcs $A(G)\triangle (A(P)\cup A(C))$ is a galaxy which 
spans $V(G)\cup x$, contradicting the maximality of $G$.
\end{proof}

We now endow ${\cal A}\cup {x}$ with a partial order structure by letting
$a\leq b$ if there exists an alternating path starting at $a$
and ending at $b$. The fact that this relation is a partial order
relies on Claim~\ref{cl:nocircuit}. Observe that $x$ is the maximum of this order.

We also construct a digraph $\cal D$ on vertex set ${\cal A}\cup {x}$ with 
all arcs $uv\ra st$ such that $us$ or $vs$ is an arc 
of $D$ (and $uv\ra x$ such that $ux$ or $vx$ is an arc 
of $D$).

\begin{claim}\label{ordered}
The pair $(\cal D,\leq )$ is an ordered digraph.
Moreover an arc of $\cal A$ is the tail of at most one backward arc and two forward arcs, and
$x$ is the tail of at most two backward arcs.
\end{claim} 

\begin{proof} The fact that the Hasse diagram of $\leq$ is contained
in $\cal D$ follows from the fact that if  $uv\leq st$
belongs to the Hasse diagram of $\leq$, there is an alternating
path starting by $uvst$, in particular, the arc $vs$ belongs 
to $D$, and thus $uv\ra st$ in $\cal D$.

Suppose that $uv\ra st$, then $vs$ or $us$ is an arc of $D$.
If $vs$ is an arc, because there is no alternating circuit, $st$ follows $uv$ on some
alternating path, and so $uv\leq st$. In this case, $uv\ra st$ is forward.\\ 
If $us$ is an arc of $D$, we claim
that $st \leq uv$. Indeed, if an alternating 
path $P$ starting at $st$ does not contain $uv$,
the galaxy with arcs $(A(G)\triangle A(P))\cup \{us\}$ spans 
$V(G)\cup x$, contradicting the maximality of $G$. In this case, $uv\ra st$ is backward.

It follows that an arc $uv$ of $\cal A$ is the tail of at most one
backward arc (since this arc and $uv$ are the two arcs leaving $u$ in $D$),
and $uv$ is the tail of at most two forward arcs (since $v$ has outdegree at most two). 
Furthermore, since $x$ has outdegree at most two, it follows that $x$ 
is the tail of at most two backward arcs.
\end{proof}

\begin{claim}
The indegree of every vertex of $\cal D$ is two.
\end{claim}

\begin{proof} Let $uv$ be a vertex of $\cal D$ which starts an
alternating path $P$. If $u$ has indegree less than two, and thus does not belong
to the set of vertices of degree four, 
the galaxy with arcs $A(G)\triangle A(P)$ spans more vertices of degree
four than $G$, a contradiction.
Let $s$ and $t$ be the
two inneighbours of $u$ in $D$. An element of ${\cal A}\cup x$ should contain  $s$, since
otherwise, the galaxy with arcs $(A(G)\triangle A(P))\cup \{su\}$ spans 
$V(G)\cup x$ and contradicts the maximality of $G$. 
Similarly an element of ${\cal A}\cup x$ contains $t$.

Observe that the same element of ${\cal A}\cup x$ cannot contain 
both $s$ and $t$ (either the arc $st$ or the arc $ts$), otherwise 
the arcs $su$ and $tu$ would be both backward or forward, which is 
impossible.
\end{proof}

At this stage, in order to apply Lemma~\ref{ordig}, we
just need to insure that the backward outdegree 
of every vertex is at most one. Since the only element
of $\cal D$ which is the tail of two backward arcs is 
$x$, we simply delete any of these two backward arcs.
The indegree of a vertex of $\cal D$ decreases by one 
but we are still fulfilling the hypothesis of Lemma~\ref{ordig}.

Hence according to this lemma, $\cal D$ contains two arcs  $\Cc\Aa$ and $\Bb\Dd$ such that $\Aa\leq \Bb\leq \Cc$,
$\Bb\leq \Dd$ and  $\Cc\not\leq \Dd$. Recall that $\Aa,\Bb,\Cc,\Dd$ are elements of 
${\cal A}\cup {x}$. In particular, there is an alternating path $P$ containing $\Aa,\Bb,\Dd$ (in this order)
which does not contain $\Cc$. Setting $\Aa=\Aa_1\Aa_2$ and $\Cc=\Cc_1\Cc_2$, note that 
the backward arc $\Cc\Aa$ corresponds to the arc $\Cc_1\Aa_1$ in $D$.
We reach a contradiction by considering the galaxy with arcs 
$(A(G)\triangle A(P))\cup \{\Cc_1\Aa_1\}$ which spans 
$V(D')\cup x$. The proof of Lemma~\ref{spanninggalaxy} is now complete.
{\hfill$\Box$ \par \vspace{11pt}

\section{Multiple Fibres}\label{+rsfibres}

In this section we consider the general problem with $n\geq 2$ fibres, and give lower and upper bounds on $\lambda_n(m,k)$.
Let us start by proving a lower bound on  $\lambda_n(m,k)$.
\begin{proposition}\label{lambdainf}
For all $m,n,k \in \mathbb N$ , we have $\ds \lambda_n(m,k)\geq 
\left\lceil\frac{m}{n}\left\lceil \frac{k}{n}\right\rceil +   \frac{k}{n} \right\rceil$
\end{proposition}

\begin{proof}
Consider the following $m$-labelled digraph $G_{n,m,k}$ with vertex set $X\sqcup Y\sqcup Z$
such that :
 
\begin{itemize}
\item $|X|=k$, $|Y|=k2^{(m+1)k}$ and $|Z|=m{|Y| \choose k}$.
\item For any $x\in X$ and $y\in Y$, there is an arc $xy$ (of whatever label).  
\item For every set $S$ of $k$ vertices of $Y$ and any integer $1\leq i\leq m$,  
there is a vertex $z^i_S$ in $Z$ which is dominated by all the vertices of $S$ via arcs
labelled $i$.
\end{itemize}

Suppose there exists an $n$-fibre
colouring of $G_{n,m,k}$ with $c < \left\lceil\frac{m}{n}\left\lceil \frac{k}{n}\right\rceil + \frac{k}{n} \right\rceil$ colours.  
For $y\in Y$ and $1\leq i\leq m$, let
$C_i(y)$ be the set of colours assigned to the arcs labelled $i$
leaving $y$.  
For $0\leq j\leq n$, let $P_j$ be 
the set of colours used on $j$ arcs entering $y$ 
(and necessarily with two different fibres).  
Then $\sum_{j=0}^n j |P_j|=k$ as $k$ arcs enter $y$.
Moreover $\sum_{j=0}^n |P_j|=c$, since $(P_0, P_1, \dots , P_n)$ is a partition of the set of colours.
Now each colour of $P_j$ may appear in at most $n-j$ of the
$C_i(y)$, so
$$\sum_{i=1}^m |C_i(y)| \leq \sum_{j=0}^n (n-j)|P_j| = n \sum_{j=0}^n |P_j| - \sum_{j=0}^n j|P_j| = cn -k.$$
Because  $|Y|>(k-1)2^{cm}$, there is a set $S$ of $k$ vertices $y$ of $Y$
having the same $m$-tuple $(C_1(y), \dots , C_m(y))=(C_1,\dots, C_m)$.
Without loss of generality, we may assume $|C_1|=\min \{|C_i| \tq 1\leq i\leq m\}$.
Hence $|C_1| \leq \frac{cn-k}{m}$.
But the vertex  $z_S^1$ has indegree $k$, so $|C_1|\geq \frac kn $.
Since $|C_1|$ is an integer, we have
$\left\lfloor \frac{cn-k}{m} \right\rfloor \geq |C_1| \geq  \left\lceil \frac kn\right\rceil$.
So $c\geq \frac{m}{n}\left\lceil \frac{k}{n}\right\rceil +   \frac{k}{n}$
. Since $c$ is an integer, we get 
$c\geq \left\lceil\frac{m}{n}\left\lceil \frac{k}{n}\right\rceil +   \frac{k}{n} \right\rceil$, a contradiction. 
\end{proof}

Note that the graph $G_{n,m,k}$ is acyclic. 
The following lemma shows that, if $m \geq n$, one cannot expect better lower bounds 
by considering acyclic digraphs. Indeed $G_{n,m,k}$ is the $m$-labelled acyclic digraph with indegree at most
$k$ for which an $n$-fibre colouring requires the more colours.

\begin{lemma}\label{upperacyclic}
Let $D$ be an acyclic $m$-labelled digraph with $\Delta^- \leq k$.
If $m\geq n$, then $\lambda_n(D)\leq \left\lceil\frac{m}{n}\left\lceil 
\frac{k}{n}\right\rceil +   \frac{k}{n} \right\rceil$.
\end{lemma}
 
\begin{proof}
Since $D$ is acyclic, its vertex set admits an ordering $(v_1, v_2, \dots ,v_p)$
such that if $v_jv_{j'}$ is an arc, then $j<j'$.

By induction on $q$,  
we shall find an $n$-fibre colouring of $D[\{v_1, \dots ,v_q\}]$ together
with sets $C_i(v_{r})$ of $\lceil \frac kn \rceil$  (potential) colours, for $1\leq i\leq m$ and $1\leq r\leq q$,  
 such that 
assigning a colour in $C_i(v_{r})$ to an arc labelled $i$ 
leaving $v_{r}$ (in the future) will fulfil the condition of
an $n$-fibre colouring at $v_{r}$.

Starting the process is easy. We may let $C_i(v_1)$'s to be any family of 
$\lceil \frac kn \rceil$-sets such that a colour appears in at most $n$ of them.

Suppose now that we have an $n$-fibre colouring of $D[\{v_1, \dots ,v_{q-1}\}]$, and that, for any $1\leq i\leq m$ 
and $1\leq r \leq q-1$, the set $C_i(v_r)$ is already determined.
Let us colour the arcs entering $v_q$.  Each of these arcs $v_{r}v_q$
may be assigned one of the $\lceil \frac kn \rceil$ colours of
$C_{l(v_{r}v_q)}(v_{r})$. Since a colour may be assigned to $n$ arcs (using different fibres)
entering $v_q$, one can assign a colour and a fibre to each such arc.
It remains to determine the sets $C_i(v_q)$, $1\leq i\leq m$.  

For $0\leq j\leq n$, let $P_j$ be the set
of colours assigned to $j$ arcs entering $v_q$.
Let $N=\sum_{i=0}^n (n-j) |P_j|$ and 
$(c_1, c_2, \dots , c_N)$ be a sequence of colours such that
each colour of $P_j$ appears exactly $n-j$ times and consecutively.  
For $1\leq i\leq m$, set $C_i(v_q)=\{c_a \tq a \equiv i \mod m\}$. 
As $n\leq m$, a colour appears at most once in each $C_i(v_q)$.
Moreover, $N=n \left\lceil\frac{m}{n}\left\lceil \frac{k}{n}\right\rceil +   
\frac{k}{n} \right\rceil -k \geq m \left\lceil \frac{k}{n}\right\rceil$.
So for $1\leq i\leq m$,   $|C_i(v_q)|\geq  \left\lceil \frac{k}{n}\right\rceil$.
\end{proof}

Lemma~\ref{upperacyclic} shows that the lower bound of Proposition~\ref{lambdainf}
is tight  for acyclic digraphs. 
In fact, we conjecture that  the lower bound remains tight for digraphs in general:
\begin{conjecture}\label{lambdasup}
\hspace*{1cm} $\ds \lambda_n(m,k) =
\left\lceil\frac{m}{n}\left\lceil 
\frac{k}{n}\right\rceil +   \frac{k}{n} \right\rceil$
\end{conjecture}

We now establish an upper bound on $\lambda_n(m,k)$ for general
digraphs.  Note that the graphs $G_{n,m,k}$ requires lots of
colours but have very large outdegree. We first give an upper bound on $\lambda_n(D)$
for $m$-labelled digraphs with bounded in- and outdegree. 
In this case, on can show that only "few"
colours are needed. This is derived from the following theorem of Guiduli. 

\begin{theorem}[Guiduli~\cite{Gui97}]\label{guiduli}
If $\Delta^-, \Delta^+\leq k$, then $dst(D)\leq k +20\log k + 84$.
Moreover, $D$ admits a directed star colouring with  
$k+20 \log k + 84$ colours such that 
for each vertex $v$, there are at most $10 \log k + 42$ colours assigned 
to its leaving arcs. 
\end{theorem}

As we will show below, Guiduli's Theorem can be extended to the following statement for $m$-labelled digraphs.

\begin{theorem}\label{Guiduli2}
Let $f(n,m,k)=\ds \left\lceil \frac{k+ (10m^2+5)\log k + 80m^2 + m  + 21}{n}\right\rceil$
and let $D$ be an $m$-labelled digraph with $\Delta^-, \Delta^+ \leq k$.
Then $\lambda_n(D)\leq f(n,m,k)$.  Moreover, $D$ admits an
$n$-fibre colouring with $f(n,m,k)$ colours such that for each
vertex $v$ and each label $l$, the number of colours assigned to the
arcs labelled $l$ and leaving $v$ is at most
$g(m,k)=\left\lceil(10m+5)\log k + 40m + 21\right\rceil$. 
\end{theorem}

 As one can notice, Theorem~\ref{Guiduli2} in the case $n=m=1$ is slightly better than Theorem~\ref{guiduli} (for $\Delta^-, \Delta^+\leq k$, Theorem~\ref{Guiduli2} gives $dst(D)\leq k +15\log k + 102$). But this is superficial and is only due to the upper bound given in Lemma~\ref{2k+1}, which is better than the upper bound $3\Delta$ used by Guiduli. Indeed, the methods are identical.

We recall the following definition: given a family of sets $\mathcal{F}=(A_i, i\in I)$, a {\it transversal} of $\mathcal{F}$ is a family of 
distinct elements $(t_i, i\in I)$ with $t_i \in A_i$ for all $i \in I$.

\begin{lemma}\label{transversal}
Let $D$ be an $m$-labelled digraph with $\Delta^-\leq k$. Suppose that
for each vertex $v$, there are $m$ disjoint lists $L_v^1, ..., L_v^m$
of $c$ colours each being a subset of $\{1,...,k+c\}$.  If for each
vertex $v$, the family $\bigl\{\:L^i_y \tq yx \in E(D) \mbox{\ \emph{and $yx$ is labelled $i$}}\:\bigr\}$ 
has a transversal, then there is a $1$-fibre colouring of $D$ with
$k+(2m^2+1)c+m$ colours such that for each vertex $v$ and each label $l$, 
at most $(2m+1)c+1$ colours are assigned to arcs labelled $l$ that leave $v$.
\end{lemma}

\begin{proof}
Using the transversal to colour the entering arcs at each vertex, we
obtain a colouring with few conflicts. Indeed there is no conflict
between arcs entering a same vertex.  So the only possible conflicts
are between an arc entering a vertex $v$ and an arc leaving $v$.
Since arcs leaving $v$ use at most $mc$ colours (those of
$L_v^1\cup ...\cup L_v^m$), there are at most $mc$ arcs entering $v$
having the same colour as an arc leaving $c$.  Removing such entering
arcs for every vertex $v$, we obtain a digraph $D'$ for which the
colouring with the $k+c$ colours is a $1$-fibre colouring.  We now
want to colour the arcs of $D-D'$ with few extra colours. Consider a
label $1\leq l
\leq m$ and let $D'_l$ be the digraph induced by the arcs of $D-D'$
labelled $l$. Then $D'_l$ has indegree at most $mc$. By 
Theorem~\ref{2k+1}, we can partition $D'_l$ in $2m.c+1$ star forests. Thus $D$ can be 
$1$-fibre coloured with $k+c+m(2mc+1)$ colours.
Moreover, in the above described colouring, arcs labelled $l$ which leave a vertex $v$
have a colour in $L_v^l$ or corresponding to one of the 
$2mc+1$ star forests of $D'_l$. So at most $(2m+1)c+1$ colours are assigned 
to arcs labelled $l$ leaving $v$.
\end{proof} 

We will also need the following theorem.
\begin{theorem}[Alon, McDiarmid and Reed~\cite{ADR92}]\label{Bruce}
Let $k$ and $c$ be positive integers with $k\geq c \geq 5 \log k
+20$. Choose independent random subsets $S_1, \dots, S_k$ of $X=\{1,
\dots, k+c\}$ as follows. For each $i$, choose $S_i$ by performing $c$
independent uniform samplings from $X$. Then the probability that
$S_1, \dots, S_k$ do not have a transversal is at most $k^{3-\frac c
2}$
\end{theorem}

\noindent {\bf Proof of Theorem~\ref{Guiduli2}.\ }
It suffices to prove the result for $n=1$. Indeed we can extend a $1$-fibre colouring satisfying
the conditions of the theorem into an $n$-fibre colouring satisfying the conditions
by replacing all the colours $qn+r$ with $1\leq r\leq n$ by the colour $q+1$ on fibre $r$.

Let $c=\left\lceil 5 \log k + 20\right\rceil$.
We can assume $k\geq mc$.
For all vertices $x$, select $mc$ different ordered elements
$e_1,e_2,\cdots,e_{mc}$ independently and uniformly. For all $1 \leq
i \leq m$, let $L_x^i=\{e_{ci+1},\cdots,e_{c(i+1)}\}$.
Each set has the same distribution a set of $c$ elements chosen 
uniformly and independently.

Let $A_x$ be the event that the family $\bigl\{\:L^i_y \tq yx \in E(D) \mbox{\ and $yx$ is labelled $i$}\:\bigr\}$ fails to have a transversal. By Theorem~\ref{Bruce},  $P(A_x)
\leq k^{3-c/2}$. Furthermore, the event $A_x$
is independent of all $A_y$ for which there is no vertex $z$ such that
both $zx$ and $zy$ are in $E(D)$. It follows that the dependency graph for these events
has degree at most $k^2$, and so we can apply Lov\'asz Local Lemma to
obtain that there exists a family of lists satisfying conditions of Lemma
\ref{transversal}. This lemma gives the desired colouring.
{\hfill$\Box$ \par \vspace{11pt}

For general digraphs, when we do not have $\Delta^-,\Delta^+ \leq k$, we may use the following trick to obtain an upper bound. Any digraph $D$ may be decomposed into an acyclic digraph $D_a$ and
an Eulerian digraph $D_e$ (i.e., in $D_e$, for every vertex $v$, $d_{D_e}^-(V)=d_{D_e}^+(v)$).
(To see this, consider an Eulerian subdigraph $D_e$ of $D$ which
has a maximum number of arcs. Then the digraph $D_a =D -D_e$ is necessarily acyclic.) Hence by Lemma~\ref{upperacyclic} (applied to $D_a$) and Theorem~\ref{Guiduli2} (applied to $D_e$), we have 
$$ \mbox{if } m\geq n, \mbox{ then } \lambda_n(D)\leq \left\lceil\frac{m}{n}\left\lceil 
\frac{k}{n}\right\rceil +   \frac{k}{n} \right\rceil + f(n,m,k),$$ 
for $f(n,m,k)$ the function given in Theorem~\ref{Guiduli2}. But, as we will show now, it is possible to lessen this bound by roughly $\frac{k}{n}$.

\begin{theorem}\label{fibmulti}
If $m\geq n$, then
$$\lambda_n(m,k)\leq \left\lceil\frac{m}{n}\left\lceil \frac{k}{n}\right\rceil +   
\frac{k}{n} \right\rceil + 2m\frac{\left\lceil(10m+5)\log k + 40m + 21\right\rceil}{n}.$$
\end{theorem}

\begin{proof}
Let $D$ be an $m$-labelled digraph with $\Delta^-(D) \leq k$.
Consider a decomposition of $D$ into an acyclic digraph $D_a$ and an Eulerian digraph $D_e$.
 We first apply Theorem~\ref{Guiduli2} to find an $n$-fibre colouring of the arcs of $D_e$ with $f(n,m,k)$ colours such
that, in addition, at most $g(m,k)$ colours are assigned to the arcs leaving
each vertex.

\noindent We shall extend  the $n$-fibre colouring of $D_e$ to  
the arcs of $D_a$ in a way similar to the
proof of Lemma~\ref{upperacyclic}. I.e., we will assign to 
each vertex $v$, sets $C_i(v)$, $1\leq i\leq m$   
of $\lceil \frac kn +mg(m,k)\rceil$ colours such that an arc labelled $i$ leaving $v$ 
will be labelled using a colour in $C_i(v)$.

\noindent Let $(v_1, \dots ,v_n)$ be an ordering
of the vertices of $A$ such that if $v_jv_{j'}$ is an arc then $j<j'$.
We start to build the $C_i(v_1)$ with the colours assigned to the
leaving arcs of $v_1$ labelled $i$.  The vertex $v_1$ has at most $k$
entering arcs. Each of them forbid one {\it type} (colour, fibre). In the
colouring of $D_e$ induced by Theorem~\ref{Guiduli2}, there are at most $mg(m,k)$
types assigned to the arcs leaving $v_1$. So there are at least
$\left\lceil \frac{mk}{n^2} \right\rceil + \left\lceil \frac{k}{n}
\right\rceil + 2m\frac{g(m,k)}{n}-k-mg(m,k) \geq m\left\lceil
\frac{k}{n}\right\rceil+mg(m,k)$ types unused at vertex $v_1$. Since
$m\geq n$, we can partition these types into $m$ sets of size at least
$\frac{k}{n}$ such that no two types having the same colour are in the
same set. These sets are the $C_i(v_1)$.

Suppose that the sets have been defined for $v_1$ up to $v_{q-1}$, and that all the arcs $v_iv_j$ for $i<j<q$ have been assigned a colour. We now give a colour to each arc $v_iv_q$ for $i<q$.

There are $k_e$ arcs entering $v_q$ in $D_e$ which are already
coloured. So it remains to give a colour to $k_a \leq k - k_e$ arcs.
Each uncoloured arc may be assigned a colour in a list of size at
least $\left\lceil \frac{k}{n}+mg(m,k)\right\rceil$. This gives a
choice between $n\left\lceil \frac{k}{n}+mg(m,k)\right\rceil$
different types. $k_e$ types are forbidden by the entering arcs in $D_e$ while
at most $mg(m,k)$ types are forbidden by the leaving arcs in $D_e$. 
Hence, it remains at least
$n\left\lceil \frac{k}{n}+mg(m,k)\right\rceil - k_e -mg(m,k)\geq
k_a$ types for the entering arcs of $D_a$. 
So one can assign distinct available colours to each of the $k_a$ arcs entering $v_q$. We then build the $C_i(v_q)$ as we did for $v_1$.

\noindent Once this process is finished, we obtain an $n$-fibre colouring of $D$ using 
$\left\lceil \frac{mk}{n^2} \right\rceil + \left\lceil \frac{k}{n} \right\rceil + 2m\frac{g(m,k)}{n}$ colours. 

\end{proof}

Theorem~\ref{fibmulti} gives an upper bound on $\lambda_n(m,k)$ when $m\geq n$.
We now give an upper bound for the case $m<n$.

\begin{proposition}\label{upper}
If $m < n$ then $\lambda_n(m,k)\leq \left\lceil\frac{k}{n-m}\right\rceil$.
\end{proposition}

\begin{proof}Let $D$ be an $m$-labelled digraph with $\Delta^- \leq k$.
We should show the existence of a proper $n$-fibre
colouring with $\left\lceil\frac{k}{n-m}\right\rceil$. 
For each vertex $v$, we give to each of its entering arcs a colour such that
none of the colours is used more than $n-m$ times. This is possible 
since there are at most $k\leq (n-m)\left\lceil\frac{k}{n-m}\right\rceil$ arcs entering $v$. 
So we now have  $in(v,\lambda) \leq
n-m$. 
Moreover each arc $vw$ is given a colour by $w$. Since $D$ is
$m$-labelled, a colour $\lambda$ can be used to colour an arc of at
most $m$ different labels, i.e., $out(v,\lambda) \leq m$.  Consequently
$in(v,\lambda)+out(v,\lambda) \leq n$. This gives 
a proper $n$-fibre colouring.
\end{proof}

\section{Concluding Remarks}\label{sec:conclusion}
 
One question arising naturally from the previous sections is the complexity of calculating $\lambda_n(D)$ for an $m$-labelled digraph $D$. As we will show in the first subsection, unsurprisingly, this problem is ${\cal NP}$-hard even for the simpler problem of directed star arboricity and even for restricted class of digraphs of in- and outdegree bounded by two.
 We end this section by showing how a similar approach to the one in Section~\ref{seccubic} allows us to give a very short proof of a recent result of Pinlou and Sopena~\cite{PiSo04}.

\subsection{Complexity}\label{complexity}

The digraphs with directed star arboricity one are the galaxies, so
one can decide in polynomial time if $dst(D)=1$.  Deciding whether $dst(D)=2$
or not is also easy since we just 
have to check that the conflict graph (with vertex set the arcs 
of $D$, two distinct arcs $xy,uv$ being in conflict when $y=u$ or $y=v$)
is bipartite. However for larger values, as expected, it is ${\cal NP}$-complete to 
decide if a digraph has directed star arboricity at most $k$. This is illustrated
by the next result:

\begin{theorem}
The following decision problem is ${\cal NP}$-complete:\\
\emph{{\sc Instance}: A digraph $D$ with $\Delta^+(D)\leq 2$ and $\Delta^-(D)\leq 2$.}\\
\emph{{\sc Question}: Is $dst(D)$ at most $3$?}
\end{theorem}

\begin{proof}
The proof is by a reduction from $3$-edge-colouring of $3$-regular graphs, which is known to be ${\cal NP}$-complete.

\noindent Let $G$ be a $3$-regular graph. It is easy to see that $G$ admits an orientation $D$ such that every vertex has in- and outdegree at least one (i.e., $D$ does not have neither sink nor source).

Let $D'$ be the digraph obtained from $D$ by replacing every vertex with indegree one 
and outdegree two by the subgraph $H$ depicted in Figure~\ref{gadget} which has also
one entering arc (namely $\vec a$) and two leaving arcs ($\vec b$ and $\vec c$).
\begin{figure}[!hbt]
\begin{center}
\epsfig{file=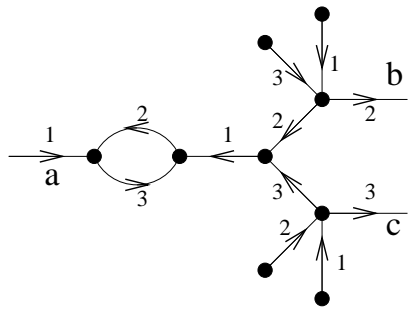, height=4cm}
\caption{The graph $H$ and one of its directed star $3$-colouring}\label{gadget}
 \end{center}
\end{figure}
It is quite easy to check that in any directed star $3$-colouring of $H$, the three arcs
$\vec a$, $\vec b$ and $\vec c$ get different colours.
Moreover, if these three arcs are precoloured with three different colours, we 
can extend this to a directed star $3$-colouring of $H$. Such a colouring with $\vec a$ 
coloured $1$, $\vec b$ coloured $2$ and $\vec c$ coloured $3$ is given in
Figure~\ref{gadget}.
Furthermore, in a directed star $3$-colouring, a vertex with indegree two and outdegree one must have its three incident arcs
coloured differently. So $dst(D')=3$ if and only if $G$ is $3$-edge colourable.
\end{proof}

\subsection{Acircuitic Directed Star Arboricity}

A directed star colouring is {\it acircuitic} if there is no {\it
bicoloured} circuits, i.e.,  circuits for which only two colours
appear on its arcs.  The {\it acircuitic directed star arboricity} of
a digraph $D$ is the minimum number $k$ of colours such that there
exists an acircuitic directed star $k$-colouring of $D$. 

In this last section, we give a short alternative proof of the following theorem.

\begin{theorem}[Pinlou and Sopena~\cite{PiSo04}]\label{pinlousopena}
Every subcubic oriented graph has acircuitic directed star arboricity at most $4$.
\end{theorem}

Indeed, it is possible to apply our Theorem~\ref{degre3} directly to derive this theorem. However, there is a shorter proof using the following lemma.

\begin{lemma}\label{3acyclic}
Let $D$ be an acyclic subcubic digraph.
Let $L$ be a list-assignment on the arcs of $D$
such that for every arc $uv$, $|L(uv)|\geq d(v)$.
Then $D$ admits a directed star $L$-colouring.
\end{lemma}
\begin{proof}
We prove the result by induction on the number of arcs of $D$, the result holds
trivially if $D$ has no arcs.

\noindent Since $D$ is acyclic, it has an arc $xy$ with $y$ a sink.
Let $\omega$ be a colour in $L(xy)$.
For any arc $\vec a$ distinct from $xy$,
set $L'(\vec a)=L(\vec a)\setminus \{\omega\}$ if $\vec a$ is incident to $xy$ 
(and thus has head in $\{x,y\}$ since $y$ is a sink),
and  $L'(\vec a)=L(\vec a)$ otherwise.
Then in $D'=D-xy$, we have $|L'(uv)|\geq d(v)$ for any arc $uv \neq xy$.
Hence, by induction hypothesis, $D'$ admits a directed star $L'$-colouring that can be extended
to a directed star $L$-colouring of $D$ by colouring $xy$ with $\omega$.
\end{proof}

\noindent {\bf Proof of Theorem~\ref{pinlousopena}.\ }

Let $V_1$ be the set of vertices of outdegree at most one and
$V_2 =V\setminus V_1$. Every vertex of $V_2$ has outdegree at least two (and so indegree at most one).

\noindent Let $M$ be the set of arcs with tail in $V_1$ and head in $V_2$.
We colour all the arcs of $M$ with colour $4$. Moreover for every circuit $C$ in $D[V_1]$ or in $D[V_2]$, we choose an arc $\vec a(C)$ and colour it by $4$.
Note that, by definition of $V_1$ and $V_2$, the arc $\vec a(C)$ is not incident to any arc of $M$, and in addition, 
$C$ is the unique circuit containing $\vec a(C)$.
Let us denote by $M_4$ the set of all arcs coloured by $4$.
It is easily seen that $M_4$ is a matching and $D-M_4$ is acyclic.

\noindent We shall now find a directed star colouring of $D-M_4$ with colours $\{1,2,3\}$
that does not create any bicoloured circuit.
In any colouring of the arcs, if such a circuit existed, $4$ would be one of 
its colour because $D-M_4$ is acyclic, and moreover,
all the arcs of this circuit coloured by $4$ would be in $M$, because each arc in  $M_4\setminus M$ is in a unique circuit and this unique circuit has a unique arc coloured by $4$. Hence we just have to be careful when dealing with arcs in the digraph induced by the endvertices of the arcs
of $M$.

\noindent Let us denote the arcs of $M$ by $x_iy_i$, $1\leq i\leq p$, and  
set $X=\{x_i, 1\leq i\leq p\}$ and  $Y=\{y_i, 1\leq i\leq p\}$ (we have then $x_i\in V_1$ and $y_i\in V_2$).
Let $E'$ be the set of arcs with tail in $Y$ and head in $X$.
Let $H$ be the graph with vertex set $E'$ such that
an arc $y_ix_j$ is adjacent to an arc $y_kx_l$ if
\begin{itemize}
\item[$(a)$] Either $k=l$,
\item[$(b)$] Or $j=k$ and $i>j$ and $l>j$.
\end{itemize}
Since a vertex of $X$ has indegree at most two and a vertex of $Y$ has
outdegree at most two, $H$ has maximum degree three. Moreover $H$ contains no 
$K_4$, because two arcs of $E'$ with same tail $y_k$ are not 
adjacent in $H$. Hence, by Brooks Theorem, $H$ has a vertex-colouring with colours
$\{1,2,3\}$, and this colouring corresponds to a colouring $c$ of the arcs of
$E'$. Since $(a)$ is satisfied, $c$ is a directed star colouring.
Moreover, this colouring creates no bicoloured circuits. Indeed, a
circuit contains a subpath $y_ix_jy_jx_l$, with $i>j$ and $k>j$, whose three arcs are coloured differently by $(b)$.

\noindent Let $D'=D-(M_4\cup E')$.
For any arc $uv$ in $D'$, let $L(uv)=\{1,2,3\}
\setminus \{c(wv) \ | \ wv\in E'\}$.  The set $L(uv)$ is the set of
colours in $\{1,2,3\}$ that may be assigned to $uv$ without creating any conflict
with the already coloured arcs. The digraph $D'$ is acyclic and $|L(uv)|\geq d(v)$, so 
by Lemma~\ref{3acyclic}, it admits a directed star $L$-colouring. We infer that
$D$ has an acircuitic directed star colouring with colours in $\{1,2,3,4\}$ and the theorem follows.
In addition, we note that in this colouring, the arcs coloured by $4$ form a matching.
{\hfill$\Box$ \par \vspace{11pt}

\noindent{\bf Acknowledgement.} The authors are grateful to J-C. Bermond for bringing out their attention to the paper by Brandt and Gonzalez. They thank an anonymous referee for the remarks which helped them to improve the presentation of the paper.


\end{document}